\documentclass[11pt]{article}
\usepackage[margin=1in]{geometry}
\usepackage{amsmath,amssymb,amsthm,mathtools}
\usepackage{graphicx}
\usepackage{tikz}
\usetikzlibrary{arrows.meta,calc}
\usepackage{physics}
\usepackage{bm}
\usepackage{microtype}
\usepackage[dvipsnames]{xcolor}
\usepackage[hypertexnames=false,colorlinks=true,linkcolor=MidnightBlue,citecolor=ForestGreen,urlcolor=BrickRed]{hyperref}
\usepackage{aliascnt}
\usepackage[nameinlink]{cleveref}

\numberwithin{equation}{section}
\newtheorem{theorem}{Theorem}[section]
\newcommand{\sharedtheorem}[2]{%
  \newaliascnt{#1}{theorem}%
  \newtheorem{#1}[#1]{#2}%
  \aliascntresetthe{#1}}
\sharedtheorem{proposition}{Proposition}
\sharedtheorem{lemma}{Lemma}
\sharedtheorem{corollary}{Corollary}
\sharedtheorem{remark}{Remark}
\crefname{theorem}{Theorem}{Theorems}
\crefname{proposition}{Proposition}{Propositions}
\crefname{lemma}{Lemma}{Lemmas}
\crefname{corollary}{Corollary}{Corollaries}
\crefname{remark}{Remark}{Remarks}

\title{On Optimal Quantum Data Hiding and Maximal Separable Ball}
\author{Zhi Li\\{\normalsize IBM Research}}
\date{}

\begin{document}
\maketitle

\begin{abstract}
Quantum data hiding asks how much distinguishing power can be lost when global measurements are restricted to local measurements and classical communication. In this work, we establish sharp results across the hierarchy of restricted measurements.
For bipartite systems on $\mathbb C^n\otimes\mathbb C^m$, we prove that the optimal data-hiding ratios against separable and LOCC measurements are both $\min\{n,m\}$. 
This result follows from a stronger result that, for every $2\le p\le\infty$, the largest centered Schatten $p$-ball whose associated binary measurements are implementable by finite-round LOCC has radius $\min\{n,m\}^{2/p-1}$. This strengthens the classic separable-ball theorems, while also providing an explicit finite-round LOCC implementation.
For local operations without communication (and hence also for one-way LOCC), we prove that the optimal ratio is $(1+o(1))\min\{n,m\}$, matching the asymptotic constant $1$.
\end{abstract}

\setcounter{tocdepth}{2}
\tableofcontents

\section{Introduction and main results}\label{sec:introduction}

The term quantum data hiding originates in the work of Terhal, DiVincenzo, and Leung \cite{TDL01,DLT02}, who showed that classical information can be encoded into bipartite quantum states that are readily distinguishable by a joint measurement but difficult to distinguish by local operations and classical communication (LOCC). 
Quantitatively, this phenomenon is a separation between global and locally constrained state discrimination. 
In this work, we ask how large this separation can be at fixed local dimensions.

Although quantum data hiding was originally formulated for LOCC, the same question can be asked for any prescribed class $\mathcal M$ of POVMs. Let $M_k$ be the algebra of operators on $\mathbb C^k$. 
Following the distinguishability-norm framework introduced in \cite{MWW09}, for Hermitian $h\in M_n\otimes M_m$, define
\begin{equation}\label{eq:intro-measurement-norm}
  \norm{h}_{\mathcal M}
  =
  \sup_{(M_\alpha)\in\mathcal M}
  \sum_\alpha \abs{\tr(M_\alpha h)}.
\end{equation}
If $\rho$ and $\sigma$ occur with prior probabilities $p$ and $1-p$, respectively, then for $h=p\rho-(1-p)\sigma$, the quantity $\norm{h}_{\mathcal M}$ is precisely the optimal discrimination bias achievable with measurements from $\mathcal M$. For unrestricted measurements, the corresponding bias is $\norm{h}_1$ \cite{Helstrom76}. The largest relative loss under the restriction $\mathcal M$ is therefore the data-hiding ratio
\begin{equation}\label{eq:intro-data-hiding-ratio}
  R_{\mathcal M}(n,m)
  =
  \sup_{0\ne h=h^*}
  \frac{\norm{h}_1}{\norm{h}_{\mathcal M}}.
\end{equation}

The standard locality-restricted measurement classes form a natural hierarchy. Under local operations without communication (LO), Alice and Bob perform fixed local POVMs independently, and the joint measurement outcome is the pair of their local outcomes. In $\mathrm{LOCC}_{A\to B}$, Alice measures first and Bob may choose his measurement according to her outcome. Separable measurements (SEP) and positive-partial-transpose measurements (PPT) are the usual outer relaxations of LOCC, defined by requiring every POVM effect to be separable or PPT, respectively \cite{Peres96,HHH96,Rains01}.
The inclusions
$\mathrm{LO}\subseteq\mathrm{LOCC}_{A\to B}\subseteq\mathrm{LOCC}\subseteq\mathrm{SEP}\subseteq\mathrm{PPT}$ imply
\begin{equation}\label{eq:intro-ratio-chain}
  R_{\mathrm{PPT}}(n,m)
  \le R_{\mathrm{SEP}}(n,m)
  \le R_{\mathrm{LOCC}}(n,m)
  \le R_{A\to B}(n,m)
  \le R_{\mathrm{LO}}(n,m).
\end{equation}

Determining these ratios has been the subject of a sequence of increasingly sharp estimates. Put $d=\min\{n,m\}$. On the lower-bound side, the SWAP operator on embedded $d$-dimensional local subspaces gives $R_{\mathrm{SEP}}(n,m)\ge d$ \cite{LPW18}.
This witness arises from the symmetric--antisymmetric Werner-state construction \cite{Werner89} used in the original data-hiding protocol \cite{TDL01,DLT02}, and later optimized in \cite{LPW18}.

Upper bounds were developed in several steps. For balanced systems ($m=n$), Matthews, Wehner, and Winter proved that LOCC always retains a bias of order $1/d$, establishing the matching linear dependence \cite[Corollary~17]{MWW09}. Lami, Palazuelos, and Winter extended this estimate to arbitrary local dimensions and obtained $R_{\mathrm{LOCC}}(n,m)\le 2d-1$ \cite[Theorem~16]{LPW18}. Corr{\^e}a, Lami, and Palazuelos later proved $R_{\mathrm{LO}}(n,m)\le 2\sqrt2\,d$, showing that the linear upper bound persists even without classical communication \cite[Theorem~1.1]{CLP22}. Finally, the sharp comparison with product observables in \cite{LS26} improved this to $R_{\mathrm{LO}}(n,m)\le\sqrt2\,d$. Thus, before the present work, the best bounds were
\begin{equation}\label{eq:intro-previous-bounds}
  d
  \le R_{\mathrm{SEP}}(n,m)
  \le R_{\mathrm{LOCC}}(n,m)
  \le R_{A\to B}(n,m)
  \le R_{\mathrm{LO}}(n,m)
  \le \sqrt2\,d.
\end{equation}
Thus the dimension dependence was known throughout the hierarchy, but the sharp constants for SEP, LOCC, one-way LOCC, and LO remained open, except that $R_{\mathrm{SEP}}(n,n)=n$ was known \cite{MWW09,LPW18}\footnote{Note that, by using the result of \cite[Corollary 8.4]{Ando04}, the argument in \cite{MWW09} can already give $R_{\mathrm{SEP}}(n,m)=d$.}.

Our first main result determines the ratio exactly for PPT, SEP, and LOCC.
\begin{theorem}\label{thm:exact-data-hiding}
For all $n,m\ge1$,
\begin{equation}\label{eq:exact-data-hiding}
  R_{\mathrm{PPT}}(n,m)
  =R_{\mathrm{SEP}}(n,m)
  =R_{\mathrm{LOCC}}(n,m)
  =\min\{n,m\}.
\end{equation}
\end{theorem}

In particular, this theorem implies the data-hiding protocol in \cite{LPW18} using Werner states (with unequal prior distribution) is optimal.

As observed in \cite{MWW09}, the upper bound in \cref{thm:exact-data-hiding} is closely related to a separable-ball problem. A Hermitian perturbation $H\in M_n\otimes M_m$ such that $\norm{H}_\infty\leq 1$ gives rise to the binary POVM
\begin{equation}\label{eq:intro-binary-povm}
  \left\{
  \frac12(I+H),\frac12(I-H)
  \right\}.
\end{equation}
The size of the centered ball of perturbations for which this measurement is implementable by LOCC directly controls the data-hiding ratio. We determine the relevant balls for all Schatten norms.
\begin{theorem}\label{thm:schatten-balls}
Let $2\le p\le\infty$ and $H=H^*\in M_n\otimes M_m$. If
\begin{equation}\label{eq:intro-schatten-ball}
  \norm{H}_p\le d^{\frac2p-1}, \qquad d=\min\{n,m\},
\end{equation}
then the binary POVM in \cref{eq:intro-binary-povm} is implementable by a two-round LOCC protocol.
Moreover, the radius $d^{2/p-1}$ is optimal even for PPT measurements.
\end{theorem}

\Cref{thm:schatten-balls} strengthens the classic separable-ball theorem of Gurvits and Barnum \cite{GB02} and Ando \cite{Ando04}. They proved that, for arbitrary local dimensions $n$ and $m$, the operator $I+H$ where $H=H^*$ is separable whenever $\norm{H}_2\le1$ \cite{GB02} or $\norm{H}_\infty\le 1/d$ \cite{Ando04}. 
For balanced systems on $\mathbb C^d\otimes\mathbb C^d$, \cite{GB02} further determined the optimal Schatten $p$-radius $d^{2/p-1}$ for every $2\le p\le\infty$. 
Our theorem extends these results in two directions. First, for $2<p<\infty$, it removes the restriction to equal local dimensions and gives the sharp radius $d^{2/p-1}$ for arbitrary $n$ and $m$. Second, it upgrades the conclusion from separability to LOCC implementability: under the same sharp condition, the two effects $(I+H)/2$ and $(I-H)/2$ are jointly realized by an explicit two-round LOCC protocol.

We next turn to local operations without communication. The best previous upper bound was $R_{\mathrm{LO}}(n,m)\le\sqrt2\,d$ \cite{LS26}. Our next main result shows that the asymptotic constant is in fact $1$, matching that of PPT, SEP and LOCC.
\begin{theorem}\label{thm:lo-main}
There is an absolute constant $K>0$ such that, for all $n,m\ge1$ and every Hermitian $h\in M_n\otimes M_m$, setting $d=\min\{n,m\}$, we have
\begin{equation}\label{eq:intro-lo-bound}
  \norm{h}_1\le \bigl(d+K\sqrt{d}\bigr)\norm{h}_{\mathrm{LO}}.
\end{equation}
Equivalently, 
\begin{equation}\label{eq:intro-lo-main}
  d\le R_{\mathrm{LO}}(n,m)\le d+O(\sqrt d).
\end{equation}
\end{theorem}

This bound is proved by duality, pairing $h$ against a Hermitian unitary $V$, which is decomposed via a Haar-random orthonormal basis $(P_i)_{i=1}^n$ on Alice's system (assuming $n\le m$). The main term is realized by LO measurements in which Alice measures $(P_i)$ and Bob performs a $2^n$-outcome POVM depending on the basis $(P_i)$ but not on Alice's outcome. The remaining terms are bounded by further LO norm.

Combining \cref{thm:exact-data-hiding,thm:lo-main} with \cref{eq:intro-ratio-chain}, we conclude that the entire hierarchy of data-hiding ratios collapses to the same leading constant: for $d=\min\{n,m\}$,
\begin{equation}\label{eq:intro-hierarchy-collapse}
  d
  = R_{\mathrm{PPT}}(n,m)
  = R_{\mathrm{SEP}}(n,m)
  = R_{\mathrm{LOCC}}(n,m)
  \le R_{A\to B}(n,m)
  \le R_{\mathrm{LO}}(n,m)
  \le d+O(\sqrt d).
\end{equation}
Thus, although the measurement classes from $\mathrm{PPT}$ down to $\mathrm{LO}$ are strictly separated, their optimal data-hiding ratios coincide exactly from $\mathrm{PPT}$ down to (two-round) $\mathrm{LOCC}$, and match asymptotically with leading constant $1$ all the way down to LO.

\section*{Acknowledgements and AI use}

The results in \cref{sec:balls,sec:data-hiding} (PPT, SEP, LOCC) are found by the author.
The proofs in \cref{sec:lo} (LO) are found by AI (GPT 6 Astra).
AI tools are also used to draft the manuscript (later extensively revised by the author). 

In the first arxiv version of this work, the author proved $1+o(1)$ for one-way LOCC and $\frac{\pi\sqrt3}{4}+o(1)$ for LO. 
After posting the first version, the author used AI to extend the $1+o(1)$ result from one-way LOCC to LO, resulting in the proof presented here.
Several ingredients of this proof (duality, decomposition,  Haar randomized constructions on Alice side, \cref{lem:closed-diagram-bound}) have appeared in the previous one-way LOCC proof.

\section{Maximal separable and LOCC balls}\label{sec:balls}

Given $H\in M_n\otimes M_m$ ($H=H^*$) and a measurement class $\mathcal M$, we call $H$ \emph{$\mathcal M$-admissible} if the binary POVM in \cref{eq:intro-binary-povm} belongs to $\mathcal M$.
In this section, we prove \cref{thm:schatten-balls}: for every $p\in[2,\infty]$, the operator $H$ is two-round-LOCC-admissible whenever $\norm{H}_p\le d^{\frac2p-1}$.

We first consider the two endpoints $p=2$ and $p=\infty$; the remaining Schatten balls are obtained by interpolation.
For both endpoints, we first treat the case where Alice's system has dimension two by constructing an LOCC implementation of the binary measurement. We then combine these two-dimensional protocols to treat arbitrary local dimensions.

\subsection{The Hilbert--Schmidt and operator-norm balls}

We first record a two-dimensional block construction that will be used in each branch of the protocol.
\begin{lemma}\label{lem:two-by-m}
If $X\in M_m$ and $\norm{X}_\infty\le1$, then the binary POVM
\begin{equation}\label{eq:two-by-m-povm}
  \left\{
  \frac12
  \begin{pmatrix}
    I & X\\
    X^* & I
  \end{pmatrix},
  \frac12
  \begin{pmatrix}
    I & -X\\
    -X^* & I
  \end{pmatrix}
  \right\}
\end{equation}
on $\mathbb C^2\otimes\mathbb C^m$ is implementable by a (Bob-first) one-way LOCC.
\end{lemma}

The separability counterpart of this lemma was proved by Gurvits and Barnum \cite[Proposition~1]{GB02}: they showed that the above POVM is in SEP. 
\Cref{lem:two-by-m} strengthens this conclusion by realizing the two effects jointly through an LOCC protocol.

\begin{proof}
First assume $X=U$ is unitary, and write its spectral decomposition as
\begin{equation}
  U=\sum_\ell e^{i\theta_\ell}Q_\ell.
\end{equation}
Bob performs the projective measurement $(Q_\ell)_\ell$ and sends the outcome to Alice. Conditional on $\ell$, Alice measures the binary observable
\begin{equation}
  A_{\theta_\ell}
  =e^{i\theta_\ell}E_{12}+e^{-i\theta_\ell}E_{21}
\end{equation}
on $\mathbb C^2$. This observable is Hermitian and squares to the identity. The plus effect of the resulting protocol is
\begin{equation}
  \sum_\ell \frac12(I+A_{\theta_\ell})\otimes Q_\ell
  =
  \frac12
  \begin{pmatrix}
    I & U\\
    U^* & I
  \end{pmatrix},
\end{equation}
and the minus effect is its complement.

For a general contraction $X$, choose a unitary polar factor $X=U\abs{X}$ and set
\begin{equation}
  U_\pm=U\left(\abs{X}\pm i\sqrt{I-\abs{X}^2}\right).
\end{equation}
Then $U_\pm$ are unitary and $X=(U_++U_-)/2$. Shared randomness is used to select one of these unitaries, after which the preceding protocol is applied. Averaging the two resulting POVMs gives \cref{eq:two-by-m-povm}.
\end{proof}

The next proposition gives a sufficient condition for combining the two-dimensional constructions into arbitrary dimensions. 
To state the proposition, we decompose:
\begin{equation}\label{eq:block-decomposition-H}
  H=\sum_{i,j=1}^n E_{ij}\otimes H_{ij}
  \in M_n\otimes M_m,
  \qquad H=H^*.
\end{equation}
Thus $H_{ji}=H_{ij}^*$.

\begin{proposition}\label{thm:weighted-block}
Suppose that there is an entrywise nonnegative matrix $A=(a_{ij})\in M_n$ such that
\begin{equation}\label{eq:weighted-row-sums}
  \sum_j a_{ij}=1
  \qquad (1\le i\le n)
\end{equation}
and
\begin{equation}\label{eq:weighted-products}
  a_{ij}a_{ji}\ge\norm{H_{ij}}_\infty^2
  \qquad (1\le i,j\le n).
\end{equation}
Then $H$ is two-round LOCC-admissible.
\end{proposition}

\begin{proof}
The condition \cref{eq:weighted-row-sums} enables the decomposition
\begin{equation}\label{eq:weighted-decomposition}
\begin{aligned}
  I\pm H
  ={}&
  \sum_i E_{ii}\otimes(a_{ii}I\pm H_{ii})\\
  &+\sum_{i<j}
  \left[
  (a_{ij}E_{ii}+a_{ji}E_{jj})\otimes I
  \pm E_{ij}\otimes H_{ij}
  \pm E_{ji}\otimes H_{ij}^*
  \right].
\end{aligned}
\end{equation}
We realize the terms in this decomposition as branches of an LOCC protocol.

Alice first performs the POVM whose effects, indexed by a pair $(i,j)$, are
\begin{equation}\label{eq:alice-pair-povm}
  \{a_{ii}E_{ii}\}_{i=1}^n
  \cup
  \{a_{ij}E_{ii}+a_{ji}E_{jj}\}_{1\le i<j\le n}.
\end{equation}
Indeed, the total coefficient of $E_{ii}$ is $\sum_j a_{ij}=1$. One choice of Kraus operators is
\begin{equation}
  K_{ii}=\sqrt{a_{ii}}E_{ii},
  \qquad
  K_{ij}=\sqrt{a_{ij}}E_{ii}+\sqrt{a_{ji}}E_{jj}
  \quad(i<j).
\end{equation}

For a diagonal outcome $(i,i)$, condition \cref{eq:weighted-products} gives $a_{ii}\ge\norm{H_{ii}}_\infty$. 
Bob then performs the binary POVM
\begin{equation}
  \left\{
  \frac12\left(I\pm\frac{H_{ii}}{a_{ii}}\right)
  \right\},
\end{equation}
with the branch omitted when $a_{ii}=0$. Its contribution to the two global effects is
\begin{equation}
  \frac12 E_{ii}\otimes(a_{ii}I\pm H_{ii}).
\end{equation}

For an off-diagonal outcome $(i,j)$ $(i\neq j)$, we define
\begin{equation}
  Y_{ij}
  =
  \begin{cases}
    H_{ij}/\sqrt{a_{ij}a_{ji}},&a_{ij}a_{ji}>0,\\
    0,&a_{ij}a_{ji}=0.
  \end{cases}
\end{equation}
In the second case \cref{eq:weighted-products} forces $H_{ij}=0$, and in both cases $\norm{Y_{ij}}_\infty\le1$. By \cref{lem:two-by-m}, on Alice's $(i,j)$ subspace the POVM with effects
\begin{equation}
  M_{ij}^{\pm}
  =
  \frac12\left(
  E_{ii}\otimes I+E_{jj}\otimes I
  \pm E_{ij}\otimes Y_{ij}
  \pm E_{ji}\otimes Y_{ij}^*
  \right)
\end{equation}
is LOCC-implementable. 
After Alice's first operation, this branch contributes
\begin{equation}
\begin{aligned}
  (K_{ij}^*\otimes I)M_{ij}^{\pm}(K_{ij}\otimes I)
  =\frac12\left[
  (a_{ij}E_{ii}+a_{ji}E_{jj})\otimes I
  \pm E_{ij}\otimes H_{ij}
  \pm E_{ji}\otimes H_{ij}^*
  \right].
\end{aligned}
\end{equation}

Summing the diagonal and off-diagonal branches gives $\frac12(I\pm H)$ by \cref{eq:weighted-decomposition}. 
Moreover, the protocol uses at most two messages, with communication order Alice to Bob in the diagonal branches, and Alice to Bob to Alice in the off-diagonal branches.
\end{proof}

We now apply \cref{thm:weighted-block} at the two endpoints $p=2$ and $p=\infty$. 
The first recovers the Gurvits–Barnum Hilbert–Schmidt ball result (see \cite[Theorem 1]{GB02}) with the stronger conclusion of LOCC implementability. 
The second recovers Ando's result on the operator-norm ball (see \cite[Corollary 8.4]{Ando04}) with the same stronger conclusion.
\begin{theorem}\label{thm:hs-ball}
If $H=H^*$ and $\norm{H}_2\le1$, then $H$ is two-round-LOCC-admissible.
\end{theorem}

\begin{proof}
Define the real symmetric nonnegative matrix $M=(x_{ij})_{i,j=1}^n$ by
\begin{equation}
  x_{ij}=\norm{H_{ij}}_\infty.
\end{equation}
Since the operator norm is bounded by the Hilbert--Schmidt norm, we have:
\begin{equation}\label{eq:block-norm-matrix-hs}
  \norm{M}_2^2
  =\sum_{i,j}\norm{H_{ij}}_\infty^2
  \le\sum_{i,j}\norm{H_{ij}}_2^2
  =\norm{H}_2^2
  \le1.
\end{equation}
Now we apply Perron--Frobenius theory to the matrix $M$.
After a simultaneous permutation of rows and columns, $M$ is a direct sum of irreducible nonnegative matrices. Perron--Frobenius theory gives, on each block $b$, a strictly positive vector $r^{(b)}$ satisfying
\begin{equation}
  M_br^{(b)}=\rho(M_b)r^{(b)},
  \qquad
  \rho(M_b)\le\norm{M_b}_2\le1.
\end{equation}
Combining these vectors gives $r\in\mathbb R^n$ with strictly positive entries and $Mr\le r$ entrywise.

We now rescale the entries of $M$ so that the products $x_{ij}x_{ji}$ are unchanged while the row sums are controlled by $Mr\le r$. Set
\begin{equation}
  b_{ij}=x_{ij}\frac{r_j}{r_i}.
\end{equation}
Then
\begin{equation}
  b_{ij}b_{ji}=x_{ij}^2=\norm{H_{ij}}_\infty^2,
  \qquad
  \sum_jb_{ij}=\frac{(Mr)_i}{r_i}\le1.
\end{equation}
Let $\varepsilon_i=1-\sum_jb_{ij}$ and define $a_{ij}=b_{ij}+\varepsilon_i\delta_{ij}$. The matrix $A=(a_{ij})$ is nonnegative, its rows sum to one, and
\begin{equation}
  a_{ij}a_{ji}\ge b_{ij}b_{ji}=\norm{H_{ij}}_\infty^2.
\end{equation}

The conclusion of the theorem then follows from \cref{thm:weighted-block}.
\end{proof}

\begin{theorem}\label{thm:operator-ball}
If $H=H^*$ and $\norm{H}_\infty\le1/d$, then $H$ is two-round-LOCC-admissible.
\end{theorem}

\begin{proof}
We may assume $n\le m$, so $d=n$. For $i\ne j$, put $a_{ij}=\norm{H_{ij}}_\infty$, and set
\begin{equation}
  a_{ii}=1-\sum_{j\ne i}a_{ij}.
\end{equation}
Since each block is a compression of $H$,
\begin{equation}
  \norm{H_{ii}}_\infty+
  \sum_{j\ne i}\norm{H_{ij}}_\infty
  \le n\norm{H}_\infty
  \le1.
\end{equation}
Thus $a_{ii}\ge\norm{H_{ii}}_\infty\ge0$. Hence $A=(a_{ij})$ is entrywise nonnegative, and every row of $A$ sums to one. For $i\ne j$, we have $a_{ij}a_{ji}=\norm{H_{ij}}_\infty^2$, while the preceding estimate gives $a_{ii}^2\ge\norm{H_{ii}}_\infty^2$. Therefore \cref{thm:weighted-block} applies.
\end{proof}

\subsection{Schatten balls}

We finish the proof of \cref{thm:schatten-balls} by interpolating between \cref{thm:hs-ball,thm:operator-ball}.

\begin{proof}[Proof of \Cref{thm:schatten-balls}]
Let $\mathcal K$ be the set of Hermitian operators $H$ that are two-round-LOCC admissible. This set is convex, symmetric, and closed\footnote{By \cite[Corollary~3]{CLMOW14}, fixed-round LOCC instruments with fixed number of outcomes form a compact set, even if intermediate measurements may use an unbounded number of outcomes.}.
 By \cref{thm:hs-ball,thm:operator-ball},
\begin{equation}\label{eq:endpoint-ball-inclusions}
  \{H=H^*: \norm{H}_2\le1\}\subseteq\mathcal K,
  \qquad
  \{H=H^*: \norm{H}_\infty\le1/d\}\subseteq\mathcal K.
\end{equation}

Let $\mathcal K^\circ$ be the polar of $\mathcal K$ with respect to the trace pairing:
\begin{equation}
\mathcal K^\circ=\{Y=Y^*: \abs{\tr(YH)}\le1,~~\forall H\in\mathcal K\}.
\end{equation} 
The two inclusions in \cref{eq:endpoint-ball-inclusions} imply that every $Y\in\mathcal K^\circ$ satisfies
\begin{equation}\label{eq:polar-endpoint-bounds}
  \norm{Y}_2\le1,
  \qquad
  \norm{Y}_1\le d.
\end{equation}
Let $q$ be conjugate to $p$, and choose $0\le\theta\le1$ so that
\begin{equation}
  \frac1q=\frac{1-\theta}{2}+\frac{\theta}{1}.
\end{equation}
Then $\theta=1-2/p$. Log-convexity of the Schatten norms and \cref{eq:polar-endpoint-bounds} give
\begin{equation}
  \norm{Y}_q
  \le\norm{Y}_2^{1-\theta}\norm{Y}_1^\theta
  \le d^{1-\frac2p}.
\end{equation}
Consequently, if $H=H^*$ and $\norm{H}_p\le d^{2/p-1}$, then
\begin{equation}
  \abs{\tr(YH)}
  \le\norm{Y}_q\norm{H}_p
  \le1
  \qquad (\forall Y\in\mathcal K^\circ).
\end{equation}
The bipolar theorem and the closedness of $\mathcal K$ imply that $H\in\mathcal K$.
\end{proof}

The fact that $d^{2/p-1}$ is the largest possible radius is already established in \cite[Theorem~3]{GB02}.
For completeness, we restate the proof here, presenting it in a somewhat more transparent form.
\begin{proof}[Proof of the optimality in \Cref{thm:schatten-balls}]
Let $  F_d=\sum_{i,j=1}^d E_{ij}\otimes E_{ji}$ be the SWAP operator on embedded $d$-dimensional local subspaces. Then
\begin{equation}\label{eq:swap-partial-transpose}
  F_d^\Gamma
  =\sum_{i,j=1}^d E_{ij}\otimes E_{ij}
  =d\ket{\Phi_d}\!\bra{\Phi_d},
  \qquad
  \ket{\Phi_d}=\frac1{\sqrt d}\sum_{i=1}^d e_i\otimes e_i.
\end{equation}
For $t>1$, the Hermitian operator $H_t=tF_d/d$ satisfies
\begin{equation}
  \norm{H_t}_p=t d^{2/p-1},
  \qquad
  (I-H_t)^\Gamma
  =I-t\ket{\Phi_d}\!\bra{\Phi_d}
  \not\succeq0.
\end{equation}
Thus $H_t$ is not PPT-admissible.
\end{proof}

\section{Optimal data hiding against PPT, SEP and LOCC}\label{sec:data-hiding}

We now apply \cref{thm:schatten-balls} to the distinguishability norms introduced in \cref{eq:intro-measurement-norm}. Let $1\le q\le2$, and let $p$ be its conjugate exponent. For every Hermitian $W$ with $\norm{W}_p\le1$, \cref{thm:schatten-balls} shows that
\begin{equation}
  F_+=\frac12(I+d^{\frac 2p-1}W),
  \qquad
  F_-=\frac12(I-d^{\frac 2p-1}W)
\end{equation}
form a two-round LOCC POVM. Hence, for every Hermitian $h$,
\begin{equation}
\begin{aligned}
  \norm{h}_{\mathrm{LOCC}}
  &\ge \abs{\tr(F_+h)}+\abs{\tr(F_-h)}\\
  &\ge \abs{\tr((F_+-F_-)h)}
  =d^{\frac 2p-1}\abs{\tr(Wh)}.
\end{aligned}
\end{equation}
Taking the supremum over such $W$ and using the duality 
$\norm{h}_q = \sup_{W} \abs{\tr(Wh)}$ (where $W$ is Hermitian and $\norm{W}_p\leq 1$)
gives the following family of estimates.

\begin{corollary}\label{cor:locc-dominates-schatten}
For every Hermitian $h\in M_n\otimes M_m$ and every $1\le q\le2$,
\begin{equation}\label{eq:locc-dominates-schatten}
  \norm{h}_{\mathrm{LOCC}}
  \ge d^{1-\frac 2q}\norm{h}_q.
\end{equation}
In particular,
\begin{equation}\label{eq:locc-trace-bound}
  \norm{h}_1\le d\,\norm{h}_{\mathrm{LOCC}}.
\end{equation}
\end{corollary}

It remains to show that the factor $d$ cannot be improved. Let $F_d$ be the embedded SWAP operator; let $(M_\alpha)_\alpha$ be a PPT POVM; define
\begin{equation}
  \varepsilon_\alpha = \operatorname{sgn}(\tr(M_\alpha F_d)),
\end{equation}
then
\begin{equation}
  \sum_\alpha \abs{\tr(M_\alpha F_d)} 
  =\tr((\sum_\alpha \varepsilon_\alpha M_\alpha) F_d)
  \leq \norm{\sum_\alpha \varepsilon_\alpha M_\alpha^\Gamma} \norm{F_d^\Gamma}_1.
\end{equation}
By PPT, $(M_\alpha^\Gamma)$ is also a POVM, so we have 
\begin{equation}
  \norm{\sum_\alpha \varepsilon_\alpha M_\alpha^\Gamma} 
  \leq \norm{\varepsilon}_\infty \norm{\sum_\alpha M_\alpha^\Gamma}=1,
\end{equation}
hence
\begin{equation}\label{eq:swap-ppt-norm}
  \norm{F_d}_{\mathrm{PPT}}
  =\sup_{(M_\alpha)\in\mathrm{PPT}}
  \sum_\alpha\abs{\tr(M_\alpha F_d)}
  \leq \norm{F_d^\Gamma}_1
  =d.
\end{equation}
On the other hand, $\norm{F_d}_1=d^2$. Thus $R_{\mathrm{PPT}}(n,m)\ge d$. Combining this with \cref{eq:intro-ratio-chain,eq:locc-trace-bound} proves \cref{thm:exact-data-hiding}.

We note that there exists another related witness.
Define $\tilde F_d \in M_d\otimes M_{2d}$ by:
\begin{equation}
  \tilde F_d = F_d \otimes (\ketbra{0}-\ketbra{1}).
\end{equation}
Then the same argument now gives:
\begin{equation}
  \norm{\tilde F_d}_{\mathrm{PPT}} \leq 2d,
\end{equation}
while $\norm{\tilde F_d}_1=2d^2$.
Hence, this $\tilde F_d$ also witnesses $R_{\mathrm{PPT}}(d,2d)\geq d$.
The advantage of this operator is that it is traceless Hermitian, hence can be represented as the difference of two mixed states:
\begin{equation}
  \tilde F_d \propto \rho_1-\rho_2.
\end{equation}
This gives us a quantum data-hiding with equal prior distribution.

\section{Local operations without communication}\label{sec:lo}

In this section, we prove \cref{thm:lo-main}. Let $h=h^*\in M_n\otimes M_m$. We will prove that
\begin{equation}\label{eq:lo-alice-bound}
  \norm{h}_1\le \bigl(n+O(\sqrt n)\bigr)\norm{h}_{\mathrm{LO}},
\end{equation}
where the implicit constant is absolute and independent of $m$. Exchanging the roles of Alice and Bob then gives \cref{eq:intro-lo-bound}, while the lower bound $R_{\mathrm{LO}}(n,m)\ge d$ is already provided by $R_{\mathrm{PPT}}(n,m)=d$ in \cref{thm:exact-data-hiding}.

\paragraph{Overview of the proof.}
Rather than constructing an LO measurement directly from $h$, we argue by duality. Recall that $\norm{h}_1$ can be attained by a Hermitian unitary $V\in M_n\otimes M_m$ ($V=V^*$, $V^2=I$):
\begin{equation}
  \norm{h}_1=\tr(hV).
\end{equation}
On the other hand, consider a special class of LO measurements: Alice measures a projective POVM $(P_i)_{i=1}^n$ and Bob independently measures a POVM $(F_s)$ indexed by strings of signs $s\in\{-1,1\}^n$, and they assign the sign $s_i\in\{\pm1\}$ to the joint outcome $(i,s)$.
By definition \cref{eq:intro-measurement-norm}, we have:
\begin{equation}\label{eq:lo-overview-witness}
  \abs{\tr\Bigl(h\sum_{i=1}^n P_i\otimes W_i\Bigr)}\le\norm{h}_{\mathrm{LO}},
  \qquad
  W_i:=\sum_{s\in\{-1,1\}^n}s_i F_s.
\end{equation}
To prove \cref{eq:lo-alice-bound}, it therefore suffices to decompose the Hermitian unitary $V$ into an average of operators of the form $\sum_i P_i\otimes W_i$ with total weight $n+O(1)$, leaving a residual of order $\sqrt n$; see Section~\ref{sec:lo-dual}.

The construction of this decomposition proceeds in three steps:
\begin{enumerate}
\item \emph{Exact reconstruction from a random basis.} 
We sample a Haar-random orthonormal basis $(P_i)_{i=1}^n$ (each $P_i$ is a rank-one projection) on Alice's Hilbert space. 
It is simple to get a decomposition:
\begin{equation}\label{eq:lo-overview-exact}
  V=I_n\otimes \frac1n\tr_A V +\sqrt{n(n+1)}\,\mathbb E\sum_{i=1}^n P_i\otimes S_V(P_i),
\end{equation}
where
\begin{equation}
  S_V(P_i) = \sqrt{\frac{n+1}{n}}\left(\tr_A\bigl[(P_i\otimes I_m)V\bigr]-\frac1n\tr_A V\right)
\end{equation}
can be thought of as a centered version of post-measurement $V$. 
In the following we denote $S_V(P_i)$ by $S_i$ and keep in mind that it depends on $P_i$.

While the above decomposition resembles what we need, the difficulty is that, for each fixed basis $(P_i)$, Bob must realize all $n$ blocks $S_1,\ldots,S_n$ simultaneously as the operators $W_i$ in \cref{eq:lo-overview-witness} of a \emph{single} POVM $(F_s)$ (which may depend on the basis $(P_i)$, though not on Alice's outcome $i$).

\item \emph{Joint realization by a single POVM.} 
A naive choice is $F'_s=2^{-n}\bigl(I_m+\sum_i s_i S_i\bigr)$. This already satisfies $\sum_s F'_s=I_m$ and $\sum_s s_i F'_s=S_i$. 
However, $F'_s$ may fail to be positive.
To make the operators positive, we replace the linear factor in $F'_s$ by a square and consider
\begin{equation}
    K_s = \frac{1}{2^{n+1}}\Bigl(I_m+\sum_{i=1}^n s_i S_i\Bigr)^2 + \frac{1}{2^{n+1}}\Bigl(I_m-\mathbb E\sum_i S_i^2\Bigr).
\end{equation}
We can verify that $\sum_s s_i K_s = S_i$ by explicit calculation.
Meanwhile, we will show that $\mathbb E\sum_i S_i^2\preceq I_m$, making the second term and hence $K_s$ positive.
The total sum is
\begin{equation}
    M := \sum_s K_s = I_m+ \frac12\Bigl(\sum_i S_i^2-\mathbb E\sum_i S_i^2\Bigr),
\end{equation}
which has expectation $I_m$.
Setting $F_s=M^{-1/2}K_s M^{-1/2}$, we get a valid POVM $(F_s)$.

\item \emph{Concentration and residual bound.} 
For the above constructed $F_s$, we have $W_i=M^{-1/2}S_i M^{-1/2}$.
If $\sum_{i=1}^n S_i^2$ concentrates around its expectation, then $M\approx I_m$ and hence $W_i\approx S_i$.
Therefore, we may replace $S_i$ with $W_i$ in \cref{eq:lo-overview-exact}.
This gives the main LO term $\sqrt{n(n+1)}\,\mathbb E\sum_{i=1}^n P_i\otimes W_i$ and the error term $\sqrt{n(n+1)}\,\mathbb E\sum_{i=1}^n P_i\otimes(S_i-W_i)$.
Using a tedious but conceptually straightforward Haar-average calculation (\cref{sec:lo-moments,app:noncrossing-proof}), we upper bound the fluctuations $\mathbb E(M-I_m)^2$ and $\mathbb E\sum_{i=1}^n S_i(M-I_m)^2 S_i$, which eventually bounds the error term with $h$ by $O(\sqrt n)\norm{h}_{\mathrm{LO}}$.
\end{enumerate}

\subsection{Dual formulation}\label{sec:lo-dual}

We first record a useful expression for the LO norm.

\begin{proposition}\label{prop:lo-norm-formula}
For every Hermitian $h\in M_n\otimes M_m$,
\begin{equation}\label{eq:lo-norm-formula}
  \norm{h}_{\mathrm{LO}}
  =
  \sup_{(E_i),(F_s)}
  \sum_{i=1}^N
  \tr\bigl[(E_i\otimes W_i)h\bigr],
  \qquad
  W_i=\sum_{s\in\{-1,1\}^N}s_i F_s,
\end{equation}
where the supremum is over all $N\ge1$, $N$-outcome POVMs $(E_i)_{i=1}^N$ on Alice's space, and $2^N$-outcome POVMs $(F_s)_{s\in\{-1,1\}^N}$ on Bob's space.
\end{proposition}

\begin{proof}
In an LO protocol, Alice performs a finite POVM $(E_i)_{i=1}^N$ and Bob independently performs a finite POVM $(\widetilde F_\beta)_\beta$, so the joint measurement has effects $E_i\otimes \widetilde F_\beta$. Set $h_i=\tr_A[(E_i\otimes I_m)h]$. The value of this protocol on $h$ is
\begin{equation}
  \sum_{i,\beta}\abs{\tr(\widetilde F_\beta h_i)}
  =\max_{\varepsilon_{i,\beta}\in\{-1,1\}}
  \sum_{i=1}^N\tr\Bigl(h_i\sum_\beta \varepsilon_{i,\beta}\widetilde F_\beta\Bigr).
\end{equation}
Let $(\varepsilon_{i,\beta})$ be a sign choice that achieves the maximum. 
We coarse grain $(\widetilde F_\beta)$ by the value of the string $(\varepsilon_{1,\beta},\ldots,\varepsilon_{N,\beta})\in\{-1,1\}^N$ and define a $2^N$-outcome POVM for Bob
\begin{equation}
  F_s
  =\sum_{\beta:\,(\varepsilon_{1,\beta},\ldots,\varepsilon_{N,\beta})=s} \widetilde F_\beta.
\end{equation}
It indeed satisfies
\begin{equation}
  W_i := \sum_{s\in\{-1,1\}^N}s_i F_s =
  \sum_\beta \varepsilon_{i,\beta}\widetilde F_\beta.
\end{equation}

Conversely, for any $2^N$-outcome POVM $(F_s)_s$, assigning the sign $s_i\in\{-1,1\}$ to the joint outcome $(i,s)$ gives
\begin{equation}
  \sum_{i=1}^N\tr(h_i W_i)\le\sum_{i,s}\abs{\tr(F_s h_i)}\le\norm{h}_{\mathrm{LO}}.
\end{equation}

Optimizing over $(E_i)$ and $(F_s)$ proves \cref{eq:lo-norm-formula}.
\end{proof}

It is instructive to compare \cref{prop:lo-norm-formula} with Alice-first one-way LOCC. In $\mathrm{LOCC}_{A\to B}$, Bob may choose his POVM after learning $i$, so the operators $W_i$ can be arbitrary independent Hermitian contractions ($\norm{W_i}_\infty\le1$), giving\footnote{See arxiv v1 of this work.}
\begin{equation}
  \norm{h}_{A\to B}=\sup_{(E_i)}\sum_i\norm{\tr_A[(E_i\otimes I_m)h]}_1.
\end{equation}
Under $\mathrm{LO}$, no communication is allowed, so $(W_i)_{i=1}^N$ must be produced simultaneously from a single POVM $(F_s)$ on Bob's system as in \cref{eq:lo-norm-formula}.

Our goal is to compare the right-hand side of \cref{eq:lo-norm-formula} with $\norm{h}_1$ for any Hermitian $h\in M_n\otimes M_m$. 
Note that, by duality,
\begin{equation}\label{eq:lo-trace-duality}
  \norm{h}_1
  =\sup\bigl\{\abs{\tr(hV)}:V=V^*,\ \norm{V}_\infty\le1\bigr\}
  =\sup\bigl\{\tr(hV):V=V^*,\ V^2=I\bigr\}.
\end{equation}
It is therefore enough to decompose such Hermitian unitary $V\in M_n\otimes M_m$ into operators whose pairing with $h$ is controlled by $\norm{h}_{\mathrm{LO}}$.

To bound the error term in this decomposition, we recall the row and column norms from operator space theory \cite{Pisier03}. For a family $x=(x_{ab})_{a,b=1}^n\subset M_m$, define the row norm and the column norm:
\begin{equation}
  \norm{x}_{L_1[R]}=\tr\left(\sum_{a,b}x_{ab}x_{ab}^*\right)^{1/2},
  \qquad
  \norm{x}_{L_1[C]}=\tr\left(\sum_{a,b}x_{ab}^*x_{ab}\right)^{1/2}.
\end{equation}
The row-plus-column norm
\begin{equation}\label{eq:lo-RplusC-def}
  \norm{x}_{L_1[R+C]}=\inf_{x=u+v}\left(\norm{u}_{L_1[R]}+\norm{v}_{L_1[C]}\right)
\end{equation}
is dual to the $RC$ norm \cite{Pisier03}:
\begin{equation}\label{eq:lo-RcapC-def}
  \norm{y}_{L_\infty[R\cap C]}
  =\max\qty{\biggl\lVert\sum_{a,b} y_{ab} y_{ab}^*\biggr\rVert_\infty^{1/2}, \biggl\lVert\sum_{a,b} y_{ab}^*y_{ab}\biggr\rVert_\infty^{1/2}}.
\end{equation}

The following proposition shows that a suitable decomposition of $V$ gives the desired bound.

\begin{proposition}\label{prop:lo-bounded-representation}
Suppose there is a constant $\kappa>0$ such that every Hermitian unitary $V\in M_n\otimes M_m$ ($V=V^*$, $V^2=I$) can be decomposed as
\begin{equation}\label{eq:lo-dual-decomposition}
  V
  =I_n\otimes G
  +\sqrt{n(n+1)}\,\mathbb E\sum_{i=1}^n P_i\otimes W_i
  +\sqrt n\,Z,
\end{equation}
where:
\begin{enumerate}
\item $G$ is Hermitian and $\norm{G}_\infty\le1$;
\item $(P_i)_{i=1}^n$ is a random $n$-outcome POVM for Alice, and for each $(P_i)$ there is a $2^n$-outcome POVM $(F_s)_{s\in\{-1,1\}^n}$ for Bob such that $W_i=\sum_{s\in\{-1,1\}^n}s_i F_s$;
\item $Z=\sum_{a,b=1}^n E_{ab}\otimes Z_{ab}\in M_n\otimes M_m$ is Hermitian and $\norm{(Z_{ab})}_{L_\infty[R\cap C]}\le\kappa n^{-1/2}$.
\end{enumerate}
Then for every Hermitian $h\in M_n\otimes M_m$,
\begin{equation}\label{eq:lo-dual-bound-conclusion}
  \norm{h}_1
  \le\bigl(1+\sqrt{n(n+1)}+\sqrt2\,\kappa\sqrt n\bigr)\norm{h}_{\mathrm{LO}}.
\end{equation}
\end{proposition}

\begin{proof}
Fix a Hermitian $h\in M_n\otimes M_m$ and a Hermitian unitary $V\in M_n\otimes M_m$. Pairing \cref{eq:lo-dual-decomposition} with $h$ gives
\begin{equation}\label{eq:lo-pairing-three-terms}
  \tr(hV)
  =\tr\bigl((\tr_A h)G\bigr)
  +\sqrt{n(n+1)}\,\mathbb E\tr\Bigl(h\sum_{i=1}^n P_i\otimes W_i\Bigr)
  +\sqrt n\,\tr(hZ).
\end{equation}
The first term is at most $\norm{\tr_A h}_1\le\norm{h}_{\mathrm{LO}}$ (Alice does nothing and Bob measures $\tr_A h$ in its eigenbasis). For each $(P_i)$, the second term is bounded by $\norm{h}_{\mathrm{LO}}$ by \cref{eq:lo-overview-witness}, and so is its expectation.

For the third term, write $h=\sum_{a,b=1}^n E_{ba}\otimes h_{ab}$. By the noncommutative Khintchine inequality in \cite[Theorem~1.1 and (1.4)]{LS26}, for a Haar-random unitary $w\in U(n)$,
\begin{equation}\label{eq:lo-RC-khintchine-quote}
  \norm{(h_{ab})}_{L_1[R+C]}
  \le \sqrt{2n}\,\mathbb E_{w\in U(n)}\norm{\sum_{a,b=1}^n w_{ab}h_{ab}}_1.
\end{equation}
For each term in the r.h.s., note that $\sum_{a,b}w_{ab}h_{ab}=\tr_A[(w\otimes I_m)h]$, and choose a unitary $X\in M_m$ with $\norm{\tr_A[(w\otimes I_m)h]}_1=\tr\bigl(h(w\otimes X)\bigr)$. With the spectral decompositions $w=\sum_{i=1}^n e^{i\theta_i}P_i$ and $X=\sum_\beta e^{i\phi_\beta}F_\beta$,
\begin{equation}
  \norm{\sum_{a,b=1}^n w_{ab}h_{ab}}_1
  =\abs{\sum_{i,\beta}e^{i(\theta_i+\phi_\beta)}\tr\bigl(h(P_i\otimes F_\beta)\bigr)}
  \le\sum_{i,\beta}\abs{\tr\bigl(h(P_i\otimes F_\beta)\bigr)}
  \le \norm{h}_{\mathrm{LO}}.
\end{equation}
By duality between \cref{eq:lo-RplusC-def} and \cref{eq:lo-RcapC-def} together with condition (3),
\begin{equation}\label{eq:lo-Z-pairing}
  \abs{\tr(hZ)}
  =\abs{\sum_{a,b=1}^n \tr(Z_{ab}h_{ab})}
  \le \norm{(Z_{ab}^*)}_{L_\infty[R\cap C]}\,\norm{(h_{ab})}_{L_1[R+C]}
  \le \frac{\kappa}{\sqrt n}\cdot\sqrt{2n}\,\norm{h}_{\mathrm{LO}}
  =\sqrt2\,\kappa \norm{h}_{\mathrm{LO}}.
\end{equation}
Combining the three bounds in \cref{eq:lo-pairing-three-terms} and taking the supremum over $V$ via \cref{eq:lo-trace-duality} proves \cref{eq:lo-dual-bound-conclusion}.
\end{proof}

\subsection{Random basis reconstruction and Bob's POVM}\label{sec:lo-basis}

We now construct the decomposition \cref{eq:lo-dual-decomposition} explicitly and reduce the estimate on $Z$ to two moment bounds.

Let $V\in M_n\otimes M_m$ be a Hermitian unitary. Set
\begin{equation}\label{eq:lo-V-centering}
  G=\frac1n\tr_A V,
  \qquad
  Y=\frac{1}{\sqrt n}(V-I_n\otimes G),
  \qquad
  A=\tr_A(Y^2).
\end{equation}
We have $\tr_A Y=0$ and $\norm{Y}_\infty\le 2/\sqrt n$. Moreover, using $V^2=I$, explicit calculation gives
\begin{equation}\label{eq:lo-A-identity}
  A=I_m-G^2,
\end{equation}
so $0\preceq A\preceq I_m$. For $T\in M_m$, define the positive map
\begin{equation}\label{eq:lo-Phi-def}
  \Phi(T)=\tr_A\bigl[Y(I_n\otimes T)Y\bigr],
\end{equation}
which satisfies $\Phi(I_m)=A$.

Let $U\in U(n)$ be Haar random, and let $P_i=UE_{ii}U^*$ ($1\le i\le n$) be the corresponding Haar-random orthonormal basis of rank-one projections on $\mathbb C^n$. Define
\begin{equation}\label{eq:lo-Bi-Si-Q}
  S_i=\sqrt{n+1}\,\tr_A\bigl[(P_i\otimes I_m)Y\bigr],
  \qquad
  Q=\sum_{i=1}^n S_i^2.
\end{equation}
Each $S_i$ is Hermitian, and $\sum_{i=1}^n S_i=\sqrt{n+1}\,\tr_A Y=0$.

By standard Haar average calculation, we get the following lemma.
\begin{lemma}\label{lem:lo-haar-second-moment}
For every $T\in M_m$,
\begin{equation}\label{eq:lo-exact-reconstruction}
  Y=\sqrt{n+1}\,\mathbb E\sum_{i=1}^n P_i\otimes S_i,
  \qquad
  \mathbb E\sum_{i=1}^n S_i T S_i=\Phi(T),
  \qquad
  \mathbb E Q=A.
\end{equation}
\end{lemma}

\begin{proof}
Fix a Hilbert--Schmidt orthonormal basis $(H_\alpha)_{\alpha=1}^{n^2-1}$ of the real space of traceless Hermitian matrices in $M_n$, and set $f_\alpha(U,i)=\sqrt{n+1}\,\tr(P_i H_\alpha)$. 
Since $Ue_i$ is uniform on the unit sphere of $\mathbb C^n$, the standard second-moment formula gives
\begin{equation}
  \mathbb E[\tr(P_i H_\alpha)\tr(P_i H_\gamma)]=\frac{\tr(H_\alpha H_\gamma)+\tr(H_\alpha)\tr(H_\gamma)}{n(n+1)}=\frac{\delta_{\alpha\gamma}}{n(n+1)},
\end{equation}
so
\begin{equation}\label{eq:lo-f-orthonormal}
  \mathbb E\sum_{i=1}^n f_\alpha(U,i)f_\gamma(U,i)=\delta_{\alpha\gamma}.
\end{equation}

Since $Y$ is Hermitian and $\tr_A Y=0$, we can expand $Y=\sum_\alpha H_\alpha\otimes y_\alpha$ with $y_\alpha\in M_m^{\mathrm{sa}}$. Then $S_i=\sum_\alpha f_\alpha(U,i)y_\alpha$. Writing $P_i=\frac1n I_n+\sum_\alpha \tr(P_i H_\alpha)H_\alpha$ and using $\sum_i S_i=0$ together with \cref{eq:lo-f-orthonormal}, we obtain
\begin{equation}
  \sqrt{n+1}\,\mathbb E\sum_{i=1}^n P_i\otimes S_i
  =\sum_{\alpha,\gamma}\Bigl(\mathbb E\sum_{i=1}^n f_\alpha(U,i)f_\gamma(U,i)\Bigr)H_\alpha\otimes y_\gamma
  =\sum_\alpha H_\alpha\otimes y_\alpha
  =Y.
\end{equation}
Similarly, \cref{eq:lo-f-orthonormal} and $\tr(H_\alpha H_\gamma)=\delta_{\alpha\gamma}$ give
\begin{equation}
  \mathbb E\sum_{i=1}^n S_i T S_i
  =\sum_\alpha y_\alpha T y_\alpha
  =\tr_A\bigl[Y(I_n\otimes T)Y\bigr]
  =\Phi(T).
\end{equation}
Setting $T=I_m$ gives $\mathbb E Q=\Phi(I_m)=A$.
\end{proof}

We next construct Bob's POVM $(F_s)$ for a fixed Alice's basis $(P_i)_{i=1}^n$. 
For $s\in\{-1,1\}^n$, we define:
\begin{equation}\label{eq:lo-K-epsilon}
  K_s
  =\frac1{2^{n+1}}\Bigl[\Bigl(I_m+\sum_{i=1}^n s_i S_i\Bigr)^2+I_m-A\Bigr]\succeq0.
\end{equation}
Summing over $s\in\{-1,1\}^n$ cancels the linear terms and the cross terms $S_iS_j$ ($i\ne j$), so
\begin{equation}\label{eq:lo-M-sum}
  M:=\sum_{s} K_s
  =I_m+\frac12 D\succeq\frac12 I_m,
  \qquad
  \sum_{s} s_i K_s
  =S_i
  \quad (1\le i\le n).
\end{equation}
Here $D=Q-A$; $M\succeq\frac12 I_m$ follows from $Q\succeq0$ and $A\preceq I_m$. Define
\begin{equation}
  F_s=M^{-1/2}K_sM^{-1/2},
\end{equation}
then $(F_s)_s$ is a POVM depending only on $(P_i)$, and
\begin{equation}\label{eq:lo-Wi-def}
  W_i:=\sum_{s}s_i F_s
  =M^{-1/2} S_i M^{-1/2},
  \qquad
  \sum_{i=1}^n W_i=0.
\end{equation}
Finally, we define
\begin{equation}\label{eq:lo-Z-def}
  Z=\sqrt{n+1}\,\mathbb E\sum_{i=1}^n P_i\otimes (S_i-W_i).
\end{equation}
By \cref{eq:lo-exact-reconstruction}, $Y=\sqrt{n+1}\,\mathbb E\sum_{i=1}^n P_i\otimes W_i+Z$, so $V=I_n\otimes G+\sqrt n\,Y$ has the form \cref{eq:lo-dual-decomposition}.

\begin{lemma}\label{lem:lo-residual-RC}
With the notation above,
\begin{equation}\label{eq:lo-Z-RC-from-D}
  \norm{(Z_{ab})}_{L_\infty[R\cap C]}^2
  \le \norm{\mathbb E\sum_{i=1}^n S_i D^2 S_i}_\infty
  +2\norm{\mathbb E D^2}_\infty.
\end{equation}
\end{lemma}

\begin{proof}
Put $R_i=S_i-W_i$ and write $Z=\sum_{a,b=1}^n E_{ab}\otimes Z_{ab}$ as in \cref{prop:lo-bounded-representation}. Since $Z$ is Hermitian and $\tr_A Z=\sqrt{n+1}\,\mathbb E\sum_i R_i=0$, we can expand $Z=\sum_\alpha H_\alpha\otimes z_\alpha$, where, with $H_\alpha$ and $f_\alpha$ as in the proof of \cref{lem:lo-haar-second-moment}, $z_\alpha=\tr_A[(H_\alpha\otimes I_m)Z]=\mathbb E\sum_i f_\alpha(U,i)R_i$. By \cref{eq:lo-f-orthonormal}, the functions $f_\alpha$ are orthonormal in $L_2\bigl(U(n)\times\{1,\ldots,n\}\bigr)$, so Bessel's inequality gives $\sum_\alpha z_\alpha^2\preceq\mathbb E\sum_i R_i^2$. Since $Z=Z^*$ and $\tr(H_\alpha H_\gamma)=\delta_{\alpha\gamma}$,
\begin{equation}
  \sum_{a,b=1}^n Z_{ab}^*Z_{ab}
  =\sum_{a,b=1}^n Z_{ab}Z_{ab}^*
  =\tr_A(Z^2)
  =\sum_\alpha z_\alpha^2
  \preceq\mathbb E\sum_{i=1}^n R_i^2.
\end{equation}

It remains to bound $\sum_i R_i^2$. With $J=I_m-M^{-1/2}$ we have $R_i=JS_i+M^{-1/2}S_iJ$.
Using $M^{-1}\preceq 2I_m$ and the fact that $(X_1+X_2)^*(X_1+X_2)\preceq 2X_1^*X_1+2X_2^*X_2$ for any $X_1$ and $X_2$, we have
\begin{equation}
  \sum_{i=1}^n R_i^2
  \preceq 2\sum_{i=1}^n S_i J^2 S_i + 4JQJ
  \preceq \sum_{i=1}^n S_i D^2 S_i + 2D^2.
\end{equation}
Here the second inequality uses $J^2,\,JQJ\preceq 2JMJ\preceq\frac12D^2$, which follows from $I_m,Q\preceq 2M$ and $JMJ=(M^{1/2}-I_m)^2\preceq(M-I_m)^2=\frac14D^2$.
Taking expectations gives \cref{eq:lo-Z-RC-from-D}.
\end{proof}

\subsection{Moment estimates}\label{sec:lo-moments}

By \cref{prop:lo-bounded-representation,lem:lo-residual-RC}, to finish the proof of \cref{thm:lo-main} it remains to show that
\begin{equation}\label{eq:lo-fluctuation-target}
  \norm{\mathbb E D^2}_\infty
  +\norm{\mathbb E\sum_{i=1}^n S_i D^2 S_i}_\infty
  =O(n^{-1}),
\end{equation}
where the implicit constant is independent of $m$. 

Expanding $D^2=(Q-A)^2$ leads to sums such as $\sum_{i,j}S_i S_i S_j S_j$ and $\sum_{i,j,k}S_i S_j S_j S_k S_k S_i$. In every term, each index appears twice, and the paired positions (such as $(1,6)$, $(2,3)$, $(4,5)$ in the second sum) do not cross. We evaluate all such expectations using a general lemma.

Let $r=2s$ be even, and let $\pi$ be a pairing of $\{1,\ldots,r\}$. Recall that $\pi$ is \emph{noncrossing} if there are no $1\le u<v<u'<v'\le r$ with $\pi(u)=u'$ and $\pi(v)=v'$. Given $T_1,\ldots,T_r\in M_m$, define
\begin{equation}\label{eq:lo-summed-word-def}
  \mathcal W_\pi(T_1,\ldots,T_r)
  =\sum_{\substack{i_1,\ldots,i_r=1\\ i_{\pi(t)}=i_t}}^n
  S_{i_1}T_1 S_{i_2}T_2\cdots S_{i_r}T_r.
\end{equation}
Namely, we contract each paired indices of $S$.

Recall from \cref{eq:lo-exact-reconstruction} that a single pair satisfies $\mathbb E\sum_{i=1}^n S_i T S_i=\Phi(T)$. 
If the $s$ index pairs were independent, we could take the expectation of \cref{eq:lo-summed-word-def} by repeatedly using \cref{eq:lo-exact-reconstruction} (see \cref{fig:noncrossing-decoupling}).
Formally, we define $\mathcal E_\pi(T_1,\ldots,T_r)\in M_m$ as the operator obtained by repeatedly replacing an innermost pair $\sum_{i=1}^n S_i T S_i$ with $\Phi(T)$ until no $S$ factors remain (with $\mathcal E_\emptyset=I_m$). 
For example, for $\sum_{i,j,k}S_i S_j^2 S_k^2 S_i$, replacing the two inner sums $\sum_j S_j^2$ and $\sum_k S_k^2$ by $\Phi(I_m)=A$ leaves $\sum_i S_i A^2 S_i$, which then becomes $\Phi(A^2)$.

The following lemma asserts that, for such planar diagrams at large $n$ \cite{tHooft74}, the pairs in $\pi$ indeed decouple at leading order as if driven by independent Haar bases, matching the operator-valued free Wick rule \cite{Spe98}. See \cref{fig:noncrossing-decoupling} for an illustration.
The proof is deferred to \cref{app:noncrossing-proof}.\footnote{We will only use this lemma for $s\leq 3$, so it is also possible to check this approximation by a brute-force computation. However, we present the lemma in this unified form since it might be of independent interest.}
\begin{lemma}\label{lem:noncrossing-pair}
Let $r=2s$ be even. There is a constant $C_r>0$ depending only on $r$ such that, for every Hermitian unitary $V\in M_n\otimes M_m$, every noncrossing pairing $\pi$ of $\{1,\ldots,r\}$, and all $T_1,\ldots,T_r\in M_m$ with $\norm{T_t}_\infty\le1$,
\begin{equation}\label{eq:noncrossing-lemma-bound}
  \norm{\mathbb E\bigl[\mathcal W_\pi(T_1,\ldots,T_r)\bigr]-\mathcal E_\pi(T_1,\ldots,T_r)}_\infty
  \le \frac{C_r}{n}.
\end{equation}
\end{lemma}

\begin{figure}[ht]
  \centering
  \newcommand{\drawframe}[9]{%
    \begin{scope}[shift={(#1,#2)}, scale=#3]
      \fill[gray!18] (0,0) -- ({0.36*(#4)},{0.36*(#5)}) -- ({0.36*(#4+#6)},{0.36*(#5+#7)}) -- ({0.36*(#6)},{0.36*(#7)}) -- cycle;
      \fill[gray!12] (0,0) -- ({0.36*(#6)},{0.36*(#7)}) -- ({0.36*(#6+#8)},{0.36*(#7+#9)}) -- ({0.36*(#8)},{0.36*(#9)}) -- cycle;
      \fill[gray!24] (0,0) -- ({0.36*(#8)},{0.36*(#9)}) -- ({0.36*(#8+#4)},{0.36*(#9+#5)}) -- ({0.36*(#4)},{0.36*(#5)}) -- cycle;
      \draw[gray!55, line width=0.35pt] ({0.36*(#4)},{0.36*(#5)}) -- ({0.36*(#4+#6)},{0.36*(#5+#7)}) -- ({0.36*(#6)},{0.36*(#7)});
      \draw[gray!55, line width=0.35pt] ({0.36*(#6)},{0.36*(#7)}) -- ({0.36*(#6+#8)},{0.36*(#7+#9)}) -- ({0.36*(#8)},{0.36*(#9)});
      \draw[gray!55, line width=0.35pt] ({0.36*(#8)},{0.36*(#9)}) -- ({0.36*(#8+#4)},{0.36*(#9+#5)}) -- ({0.36*(#4)},{0.36*(#5)});
      \draw[-{Stealth[length=3.6pt,width=2.7pt]}, ForestGreen!85!black, line width=0.75pt] (0,0) -- (#4,#5);
      \draw[-{Stealth[length=3.6pt,width=2.7pt]}, BurntOrange!90!black, line width=0.75pt] (0,0) -- (#6,#7);
      \draw[-{Stealth[length=3.6pt,width=2.7pt]}, RoyalPurple!90!black, line width=0.75pt] (0,0) -- (#8,#9);
      \fill[black] (0,0) circle (1.0pt);
    \end{scope}%
  }
  \begin{tikzpicture}[
    x=0.88cm,
    y=0.88cm,
    sbox/.style={draw, fill=gray!8, rounded corners=1.5pt, minimum width=0.62cm, minimum height=0.52cm, inner sep=1pt, font=\small},
    bob/.style={draw=blue!65!black, line width=0.75pt},
    pair/.style={draw=red!75!black, line width=0.7pt},
    seed/.style={draw=gray!70!black, densely dashed, line width=0.55pt}
  ]
    \begin{scope}[shift={(0,0)}]
      \node[font=\small] at (3.5, 2.05) {True expectation: $\mathbb E[\mathcal W_\pi]$};

      \draw[bob] (0.25, 0) -- (6.75, 0);

      \node[sbox] (L1) at (1,0) {$S_i$};
      \node[sbox] (L2) at (2,0) {$S_j$};
      \node[sbox] (L3) at (3,0) {$S_j$};
      \node[sbox] (L4) at (4,0) {$S_k$};
      \node[sbox] (L5) at (5,0) {$S_k$};
      \node[sbox] (L6) at (6,0) {$S_i$};

      \draw[pair] (L1.north) .. controls +(0, 1.35) and +(0, 1.35) .. node[pos=0.5, above=-1pt, font=\scriptsize, text=red!75!black] {$i$} (L6.north);
      \draw[pair] (L2.north) .. controls +(0, 0.58) and +(0, 0.58) .. node[pos=0.5, above=-1pt, font=\scriptsize, text=red!75!black] {$j$} (L3.north);
      \draw[pair] (L4.north) .. controls +(0, 0.58) and +(0, 0.58) .. node[pos=0.5, above=-1pt, font=\scriptsize, text=red!75!black] {$k$} (L5.north);

      \coordinate (LU) at (3.5, -2.05);
      \foreach \k in {1,...,6} {
        \draw[seed] ($(LU)+(0,0.42)$) -- (L\k.south);
      }

      \drawframe{3.5}{-2.05}{0.62}{-0.43}{-0.69}{0.79}{0.03}{-0.45}{0.72}
      \node[font=\small, right] at (3.98, -2.12) {$U$};

      \node[font=\footnotesize, text=black!75] at (3.5, -3.02) {one random basis $U$ controls all $S$ factors};
    \end{scope}

    \node[font=\Large] at (7.75, -0.35) {$\approx$};

    \begin{scope}[shift={(8.5,0)}]
      \node[font=\small] at (3.5, 2.05) {Leading term: $\mathcal E_\pi$};

      \draw[bob] (0.25, 0) -- (6.75, 0);

      \node[sbox] (R1) at (1,0) {$S_i$};
      \node[sbox] (R2) at (2,0) {$S_j$};
      \node[sbox] (R3) at (3,0) {$S_j$};
      \node[sbox] (R4) at (4,0) {$S_k$};
      \node[sbox] (R5) at (5,0) {$S_k$};
      \node[sbox] (R6) at (6,0) {$S_i$};

      \draw[pair] (R1.north) .. controls +(0, 1.35) and +(0, 1.35) .. node[pos=0.5, above=-1pt, font=\scriptsize, text=red!75!black] {$i$} (R6.north);
      \draw[pair] (R2.north) .. controls +(0, 0.58) and +(0, 0.58) .. node[pos=0.5, above=-1pt, font=\scriptsize, text=red!75!black] {$j$} (R3.north);
      \draw[pair] (R4.north) .. controls +(0, 0.58) and +(0, 0.58) .. node[pos=0.5, above=-1pt, font=\scriptsize, text=red!75!black] {$k$} (R5.north);

      \draw[seed] (3.12, -2.18) .. controls (1.55, -2.18) and (1.05, -1.35) .. (R1.south);
      \draw[seed] (3.88, -2.18) .. controls (5.45, -2.18) and (5.95, -1.35) .. (R6.south);

      \draw[seed] (2.5, -0.78) -- (R2.south);
      \draw[seed] (2.5, -0.78) -- (R3.south);
      \draw[seed] (4.5, -0.78) -- (R4.south);
      \draw[seed] (4.5, -0.78) -- (R5.south);

      \drawframe{2.5}{-1.12}{0.52}{0.13}{-0.86}{-0.78}{0.22}{0.62}{0.46}
      \node[font=\footnotesize, right] at (2.88, -1.16) {$U_2$};

      \drawframe{4.5}{-1.12}{0.52}{0.76}{0.09}{-0.42}{-0.69}{-0.50}{0.72}
      \node[font=\footnotesize, right] at (4.88, -1.16) {$U_3$};

      \drawframe{3.5}{-2.18}{0.52}{-0.12}{0.85}{-0.62}{-0.48}{0.78}{-0.24}
      \node[font=\footnotesize, right] at (3.88, -2.24) {$U_1$};

      \node[font=\footnotesize, text=black!75] at (3.5, -3.02) {independent bases $U_1,U_2,U_3$ for each pair};
    \end{scope}
  \end{tikzpicture}
  \caption{Illustration of \cref{lem:noncrossing-pair} for $\sum_{i,j,k}S_i S_j^2 S_k^2 S_i$, where $\pi=(1\,6)(2\,3)(4\,5)$. Up to $O(n^{-1})$ error, averaging over a single Haar random basis $U$ (left) agrees with assigning an independent Haar random basis to each noncrossing pair (right).}
  \label{fig:noncrossing-decoupling}
\end{figure}

\begin{proof}[Proof of \Cref{thm:lo-main}]
Let $V\in M_n\otimes M_m$ be a Hermitian unitary, and define $G$, $(P_i)$, $(W_i)$, and $Z$ as in \cref{sec:lo-basis}. Write $\mathcal L(T)=\sum_{i=1}^n S_i T S_i$, so that $Q=\mathcal L(I_m)$ while $A=\mathbb E \mathcal L(I_m)= \Phi(I_m)$. Then the two operators inside the expectations in \cref{eq:lo-fluctuation-target} are
\begin{equation}
  D^2=\bigl(\mathcal L(I_m)-A\bigr)^2,
  \qquad
  \sum_{i=1}^n S_i D^2 S_i
  =\mathcal L\Bigl(\bigl(\mathcal L(I_m)-A\bigr)^2\Bigr).
\end{equation}
Expanding these two expressions produces a fixed number of terms of the form \cref{eq:lo-summed-word-def} with even degree $r\le6$, noncrossing pairings, and each $T_t\in\{I_m,A,A^2\}$; since $0\preceq A\preceq I_m$, all of them satisfy $\norm{T_t}_\infty\le1$. 
By definition, $\mathcal E_\pi$ approximates each term by replacing every occurrence of $\mathcal L$ with $\Phi$. The leading evaluations of the two expressions are therefore $(\Phi(I_m)-A)^2=0$ and $\Phi\bigl((\Phi(I_m)-A)^2\bigr)=0$, so \cref{lem:noncrossing-pair} gives \cref{eq:lo-fluctuation-target}.

By \cref{lem:lo-residual-RC}, \cref{eq:lo-fluctuation-target} implies $\norm{(Z_{ab})}_{L_\infty[R\cap C]}=O(n^{-1/2})$ with an absolute implicit constant. Since $\norm{G}_\infty\le\norm{V}_\infty=1$ and \cref{eq:lo-Wi-def} holds, \cref{prop:lo-bounded-representation} gives $\norm{h}_1\le\bigl(n+O(\sqrt n)\bigr)\norm{h}_{\mathrm{LO}}$, proving \cref{eq:lo-alice-bound}. Exchanging Alice and Bob gives \cref{thm:lo-main}.
\end{proof}

\appendix
\addtocontents{toc}{\protect\setcounter{tocdepth}{1}}

\section{Moment estimates (\cref{lem:noncrossing-pair}) in the LO proof}\label{app:noncrossing-proof}

Recall from \cref{sec:lo-basis} that $Y\in M_n\otimes M_m$ is Hermitian with $\tr_A Y=0$ and $\norm{Y}_\infty\le 2/\sqrt n$. Writing $Y=\sum_{a,b=1}^n E_{ab}\otimes x_{ab}$ in standard matrix units, we have $x_{ba}=x_{ab}^*$,
\begin{equation}\label{eq:app-Y-properties}
  \sum_{a=1}^n x_{aa}=0,
  \qquad
  \Phi(T)=\sum_{a,b=1}^n x_{ab}Tx_{ba},
\end{equation}
and
\begin{equation}\label{eq:A2}
  S_i=\sqrt{n+1}\,\tr_A[(P_i\otimes I_m)Y]=\sqrt{n+1}\sum_{a,b=1}^n (P_i)_{ba}x_{ab},
\end{equation}
where $P_i=UE_{ii}U^*$ ($1\le i\le n$) for a Haar-random $U\in U(n)$.

Fix an even integer $r=2s$ (assume $s\ge1$ since $s=0$ is trivial), a noncrossing pairing $\pi$ of $\{1,\ldots,r\}$, and operators $T_1,\ldots,T_r\in M_m$ with $\norm{T_t}_\infty\le1$.
We may assume $n\ge r$: indeed, $\norm{S_i}_\infty\le\sqrt{n+1}\,\norm{Y}_\infty\le 2\sqrt{(n+1)/n}\le 2\sqrt2$, so for $n<r$ both $\mathcal W_\pi(T_1,\ldots,T_r)$ and $\mathcal E_\pi(T_1,\ldots,T_r)$ have norm $O_r(1)$, which can be absorbed into $C_r/n$.

The proof expands $\mathbb E[\mathcal W_\pi(T_1,\ldots,T_r)]$ as a linear combination of the operators $\mathcal O_\sigma$ (recall \cref{eq:A2}) where 
\begin{equation}\label{eq:lo-O-sigma}
  \mathcal O_\sigma
  =\sum_{a_1,\ldots,a_r=1}^n
  x_{a_1,a_{\sigma(1)}}T_1
  x_{a_2,a_{\sigma(2)}}T_2
  \cdots
  x_{a_r,a_{\sigma(r)}}T_r,\qquad
  \sigma\in \mathfrak S_r,
\end{equation}
then bounds $\norm{\mathcal O_\sigma}_\infty$ uniformly in $m$, and identifies the leading coefficients by M\"obius inversion on noncrossing partitions.

\subsection{Weingarten expansion}

Write $b(\tau)$ for the number of orbits on $\{1,\ldots,r\}$ generated by $\tau\in \mathfrak S_r$ and $\pi$ together, and set
\begin{equation}\label{eq:lo-k-def}
  k(\tau)=s-b(\tau)\ge0.
\end{equation}
For $\sigma\in \mathfrak S_r$, let $\abs{\sigma}=r-\#\sigma$ be its transposition length, where $\#\sigma$ is the number of cycles of $\sigma$, and let $d(\tau,\sigma)=\abs{\tau^{-1}\sigma}$ be the transposition distance between $\tau,\sigma \in \mathfrak S_r$.

Since $(P_i)_{ba}=U_{bi}\overline{U_{ai}}$, the unitary Weingarten formula \cite{Collins03,CS06} gives, for all indices $i_t,a_t,a'_t\in\{1,\ldots,n\}$ ($1\le t\le r$),
\begin{equation}
  \mathbb E\Bigl[\prod_{t=1}^r (P_{i_t})_{a'_t,a_t}\Bigr]
  =\sum_{\sigma,\tau\in \mathfrak S_r}
  \Bigl(\prod_{t=1}^r \delta_{a'_t,a_{\sigma(t)}}\delta_{i_t,i_{\tau(t)}}\Bigr)
  \cdot \operatorname{Wg}_n(\tau^{-1}\sigma),
\end{equation}
where $\operatorname{Wg}_n$ is the unitary Weingarten function on $\mathfrak S_r$. 
A straightforward calculation gives
\begin{equation}\label{eq:app-word-expansion}
  \mathbb E\bigl[\mathcal W_\pi(T_1,\ldots,T_r)\bigr]
  =\sum_{\sigma\in \mathfrak S_r} c_\sigma\mathcal O_\sigma,
  \qquad
  c_\sigma=(n+1)^s\sum_{\tau\in \mathfrak S_r} n^{b(\tau)}\operatorname{Wg}_n(\tau^{-1}\sigma).
\end{equation}

By \cite{CS06}, for fixed $r$ and all $\sigma\in \mathfrak S_r$,
\begin{equation}\label{eq:app-wg-asymptotics}
  \operatorname{Wg}_n(\sigma)
  =n^{-r-\abs{\sigma}}\bigl(\operatorname{Mob}(\sigma)+O_r(n^{-2})\bigr),
\end{equation}
where $\operatorname{Mob}(\sigma)=\prod_C (-1)^{\ell(C)-1}\mathrm{Cat}_{\ell(C)-1}$ over the cycles $C$ of $\sigma$, $\ell(C)$ denotes the length of $C$, and $\mathrm{Cat}_k=\frac1{k+1}\binom{2k}{k}$ is the $k$th Catalan number. Substituting \cref{eq:app-wg-asymptotics}, $b(\tau)=s-k(\tau)$, and $(n+1)^s=n^s(1+O_r(n^{-1}))$ into \cref{eq:app-word-expansion} gives
\begin{equation}\label{eq:app-c-expansion}
  n^{k(\sigma)}c_\sigma
  =\sum_{\tau\in \mathfrak S_r} n^{k(\sigma)-k(\tau)-d(\tau,\sigma)}
  \bigl(\operatorname{Mob}(\tau^{-1}\sigma)+O_r(n^{-1})\bigr).
\end{equation}
Note that for any $\sigma,\tau\in\mathfrak S_r$, 
\begin{equation}\label{eq:app-k-lipschitz}
  \abs{k(\sigma)-k(\tau)}\le d(\tau,\sigma).
\end{equation}
(Indeed, multiplying a permutation by a transposition merges two cycles or splits one cycle into two, so it changes $b(\cdot)$, and hence $k(\cdot)$, by at most $1$.)
Hence every exponent in \cref{eq:app-c-expansion} is nonpositive, and it vanishes if and only if $d(\tau,\sigma)=k(\sigma)-k(\tau)$. Consequently,
\begin{equation}\label{eq:app-a-limit-error}
  n^{k(\sigma)}c_\sigma=a_\sigma^\infty+O_r(n^{-1}),
  \qquad
  a_\sigma^\infty
  =\sum_{\substack{\tau\in \mathfrak S_r\\ d(\tau,\sigma)=k(\sigma)-k(\tau)}}
  \operatorname{Mob}(\tau^{-1}\sigma).
\end{equation}

\subsection{Norm bound for operator contractions}

If $\sigma$ has a fixed point $t=\sigma(t)$, then the sum over $a_t$ in \cref{eq:lo-O-sigma} contains $\sum_{a_t}x_{a_t,a_t}=\tr_A Y=0$, so $\mathcal O_\sigma=0$. It therefore remains to estimate the individual operators $\mathcal O_\sigma$ for fixed-point-free $\sigma\in\mathfrak S_r$. For such $\sigma$, call $t\in\{1,\ldots,r\}$ an ascent if $\sigma(t)>t$, and set
\begin{equation}\label{eq:lo-ascent-def}
  \operatorname{asc}(\sigma)=\#\{t\in\{1,\ldots,r\}:\sigma(t)>t\}.
\end{equation}
The following bound is uniform in $m$; compare \cite{MS12} in the scalar case.

\begin{lemma}\label{lem:closed-diagram-bound}
For every fixed-point-free $\sigma\in \mathfrak S_r$,
\begin{equation}\label{eq:lo-O-sigma-ascent-bound}
  \norm{\mathcal O_\sigma}_\infty
  \le 2^r n^{s-\operatorname{asc}(\sigma)}.
\end{equation}
\end{lemma}

We first illustrate the proof idea by an example. Take $r=4$, $T_t=I_m$, and $\sigma=(1\,4\,3\,2)$, that is,
\begin{equation}\label{eq:app-example-sigma}
  \sigma(1)=4,\qquad \sigma(2)=1,\qquad
  \sigma(3)=2,\qquad \sigma(4)=3.
\end{equation}
Then
\begin{equation}\label{eq:app-example-O}
  \mathcal O_\sigma
  =\sum_{a_1,\ldots,a_4=1}^n
  x_{a_1,a_4}x_{a_2,a_1}x_{a_3,a_2}x_{a_4,a_3}.
\end{equation}
We may view it as a matrix-product-like contraction as in \cref{fig:closed-diagram-contraction}. Bob indices of consecutive $x$s are contracted horizontally (hidden in \cref{eq:app-example-O}), while the permutation $\sigma$ prescribes the contractions of the Alice indices $a_1,\ldots,a_4$.

\begin{figure}[ht]
  \centering
  \begin{tikzpicture}[
    x=1.3cm,
    y=1cm,
    tensor/.style={draw, fill=gray!8, minimum width=1.45cm,
      minimum height=.85cm, inner sep=2pt},
    wire/.style={line width=.6pt},
    alice/.style={wire, draw=red!75!black},
    bob/.style={wire, draw=blue!60!black},
    cut/.style={draw=gray!55, densely dashed, line width=.4pt},
    boundary/.style={font=\scriptsize, align=center}
  ]
    \foreach \x/\t in {-1.05/0,1.1/1,3.3/2,5.5/3,7.65/4} {
      \draw[cut] (\x,1.65) -- (\x,-1.55);
      \node[below, font=\small] at (\x,-1.55) {$K_{\t}$};
    }

    \draw[->, wire] (7.8,2.0) -- node[above, font=\footnotesize]
      {right-to-left contraction} (-1.15,2.0);

    \draw[bob] (-1.55,0) -- (8.2,0);
    \node[boundary, below] at (-1.55,0) {Bob output};
    \node[boundary, below] at (8.2,0) {Bob input};

    \node[tensor] (Y1) at (0,0) {$x_{a_1,a_4}$};
    \node[tensor] (Y2) at (2.2,0) {$x_{a_2,a_1}$};
    \node[tensor] (Y3) at (4.4,0) {$x_{a_3,a_2}$};
    \node[tensor] (Y4) at (6.6,0) {$x_{a_4,a_3}$};

    \draw[alice] (Y1.north)
      .. controls +(.55,.7) and +(-.55,-.7) ..
      node[pos=.5, above, font=\small] {$a_1$} (Y2.south);
    \draw[alice] (Y2.north)
      .. controls +(.55,.7) and +(-.55,-.7) ..
      node[pos=.5, above, font=\small] {$a_2$} (Y3.south);
    \draw[alice] (Y3.north)
      .. controls +(.55,.7) and +(-.55,-.7) ..
      node[pos=.5, above, font=\small] {$a_3$} (Y4.south);
    \draw[alice] (Y1.south)
      .. controls +(1.1,-1.25) and +(-1.1,1.25) ..
      node[pos=.5, below, font=\small] {$a_4$} (Y4.north);
  \end{tikzpicture}
  \caption{The tensor network for $\mathcal O_\sigma$ for the $\sigma$ in \cref{eq:app-example-sigma}, with contraction proceeding from right to left. The dashed lines indicate the cuts $K_0,\ldots,K_4$.}
  \label{fig:closed-diagram-contraction}
\end{figure}

We estimate the tensor network by moving a vertical cut from right to left, across one block at a time. This expresses $\mathcal O_\sigma$ as a composition of four linear maps. Note that the domain and codomain at each step depend on the Alice and Bob index lines crossing the cut, and can vary from cut to cut.

To write the corresponding maps explicitly, for a set $\mathcal A$ of Alice labels let
\begin{equation}
  \mathcal H_{\mathcal A}=\ell_2([n]^{\mathcal A}),
  \qquad
  \mathcal H_\emptyset=\mathbb C.
\end{equation}
Each cut also crosses the Bob line, which contributes a factor $\mathbb C^m$. Moving the cut in \cref{fig:closed-diagram-contraction} from right to left therefore gives the chain
\begin{equation}\label{eq:app-example-contraction-chain}
  \mathbb C^m
  \xrightarrow{L_4}
  \mathcal H_{\{3,4\}}\otimes\mathbb C^m
  \xrightarrow{L_3}
  \mathcal H_{\{2,4\}}\otimes\mathbb C^m
  \xrightarrow{L_2}
  \mathcal H_{\{1,4\}}\otimes\mathbb C^m
  \xrightarrow{L_1}
  \mathbb C^m.
\end{equation}
Here $L_t$ inserts the factor in position $t$, and the composition $L_1L_2L_3L_4$ is $\mathcal O_\sigma$.

Now we prove \cref{lem:closed-diagram-bound} formally. We first record the operator estimates used in the contraction argument.

\begin{lemma}\label{lem:block-estimates-lo}
With $Y=\sum_{a,b=1}^nE_{ab}\otimes x_{ab}$ as above,
\begin{equation}\label{eq:block-estimates-lo}
\begin{gathered}
  \norm{Y}_\infty\le\frac2{\sqrt n},
  \qquad
  \sum_{a,b}x_{ab}^*x_{ab}
  =\sum_{a,b}x_{ab}x_{ab}^*
  =A\preceq I_m,
  \qquad
  \norm{Y^\Gamma}_\infty\le\sqrt n,
\end{gathered}
\end{equation}
where $Y^\Gamma=\sum_{a,b}E_{ab}\otimes x_{ba}$ is the partial transpose.
\end{lemma}

\begin{proof}
The first estimate was recorded in \cref{sec:lo-basis}. The two square-sum identities follow from $x_{ba}=x_{ab}^*$ and $\tr_A(Y^2)=A$, and $A\preceq I_m$ by \cref{eq:lo-A-identity}.

For the last estimate, let $\xi=\sum_a e_a\otimes\xi_a$ and $\eta=\sum_b e_b\otimes\eta_b$ be unit vectors. Since the $(a,b)$ block of $Y^\Gamma$ is $x_{ba}$, the Cauchy--Schwarz inequality gives
\begin{equation}
  \abs{\langle\xi,Y^\Gamma\eta\rangle}
  \le\sum_{a,b}\norm{\xi_a}\,\norm{x_{ba}\eta_b}
  \le\Bigl(n\sum_{a,b}\norm{x_{ba}\eta_b}^2\Bigr)^{1/2}
  =\Bigl(n\sum_b\Bigl\langle\eta_b,\sum_a x_{ba}^*x_{ba}\,\eta_b\Bigr\rangle\Bigr)^{1/2}
  \le\sqrt n,
\end{equation}
because $\sum_a x_{ba}^*x_{ba}\preceq A\preceq I_m$ for every $b$. Taking the supremum proves the claim.
\end{proof}

\begin{proof}[Proof of \Cref{lem:closed-diagram-bound}]
Fix a fixed-point-free permutation $\sigma\in\mathfrak S_r$. We first estimate a single step in the contraction. The step first applies $T_t$ on the Bob line, which has norm at most $\norm{T_t}_\infty\le1$, and then inserts the factor $x_{ab}$, where $a=a_t$ and $b=a_{\sigma(t)}$ are its row and column indices. All other index variables are carried through unchanged. After reordering tensor factors, the full map is therefore the tensor product of the effective map below with an identity operator, which does not change its operator norm.

\begin{itemize}
\item If $b$ occurs in the domain and $a$ in the codomain, the effective map is
\begin{equation}
  (Lw)_a=\sum_b x_{ab}w_b.
\end{equation}
This is the block matrix $Y=(x_{ab})$, and hence $\norm{L}_\infty\le 2n^{-1/2}$ by \cref{lem:block-estimates-lo}.

\item If $a$ occurs in the domain and $b$ in the codomain, the effective map is
\begin{equation}
  (Lw)_b=\sum_a x_{ab}w_a.
\end{equation}
Its block matrix is $Y^\Gamma$, so $\norm{L}_\infty\le \sqrt n$ by \cref{lem:block-estimates-lo}.

\item If both $a$ and $b$ occur in the codomain, the effective map is
\begin{equation}
  \mathcal C:\mathbb C^m\longrightarrow\ell_2([n]^2)\otimes\mathbb C^m,
  \qquad (\mathcal Cz)_{a,b}=x_{ab}z.
\end{equation}
Since $\mathcal C^*\mathcal C=\sum_{a,b}x_{ab}^*x_{ab}\preceq I_m$, we have $\norm{\mathcal C}_\infty\le1$.

\item If both $a$ and $b$ occur in the domain, the effective map is
\begin{equation}
  \mathcal D:\ell_2([n]^2)\otimes\mathbb C^m\longrightarrow\mathbb C^m,
  \qquad \mathcal D(w)=\sum_{a,b}x_{ab}w_{a,b}.
\end{equation}
Here $\mathcal D\mathcal D^*=\sum_{a,b}x_{ab}x_{ab}^*\preceq I_m$, and therefore $\norm{\mathcal D}_\infty\le1$.
\end{itemize}

We now apply these four estimates while moving the cut from right to left. For $0\le t\le r$, let $\mathcal M_t$ be the set of Alice labels whose index lines cross the cut with the factors in positions $t+1,\ldots,r$ to its right, and let
\begin{equation}
  K_t:\mathbb C^m
  \longrightarrow\mathcal H_{\mathcal M_t}\otimes\mathbb C^m
\end{equation}
be the contraction of this part of the diagram. Thus $\mathcal M_r=\mathcal M_0=\emptyset$, $K_r$ is the identity, and $K_0=\mathcal O_\sigma$. Moving the cut past $x_{a_t,a_{\sigma(t)}}T_t$ defines
\begin{equation}
  L_t:\mathcal H_{\mathcal M_t}\otimes\mathbb C^m
  \longrightarrow\mathcal H_{\mathcal M_{t-1}}\otimes\mathbb C^m
\end{equation}
and gives
\begin{equation}\label{eq:app-partial-contraction-recurrence}
  K_{t-1}=L_tK_t.
\end{equation}

It remains to identify which of the row and column indices occur in the domain of $L_t$. For the column label $\sigma(t)$, we have
\begin{equation}\label{eq:app-column-in-domain}
  \sigma(t)\in \mathcal M_t
  \quad\Longleftrightarrow\quad
  \sigma(t)>t,
\end{equation}
because its other occurrence is the row index in position $\sigma(t)$. Similarly,
\begin{equation}\label{eq:app-row-in-domain}
  t\in\mathcal M_t
  \quad\Longleftrightarrow\quad
  \sigma^{-1}(t)>t,
\end{equation}
since the other occurrence of $a_t$ is the column index in position $\sigma^{-1}(t)$.

Set
\begin{equation}
  e_t
  =\mathbf1_{\{\sigma(t)<t\}}
  +\mathbf1_{\{\sigma^{-1}(t)>t\}}.
\end{equation}
The first indicator is $1$ exactly when the column index is absent from the domain of $L_t$, and the second is $1$ exactly when the row index is present. The four possibilities and the corresponding local estimates are
\begin{equation}
\begin{array}{c|c|c|c}
\text{column index in domain} & \text{row index in domain}
& e_t & \text{bound for }\norm{L_t}_\infty\\
\hline
\text{yes} & \text{no}  & 0 & 2n^{-1/2}\\
\text{no}  & \text{yes} & 2 & n^{1/2}\\
\text{no}  & \text{no}  & 1 & 1\\
\text{yes} & \text{yes} & 1 & 1.
\end{array}
\end{equation}
In particular, in every case,
\begin{equation}
  \norm{L_t}_\infty\le 2n^{(e_t-1)/2}.
\end{equation}

The first indicator in $e_t$ sums to $\#\{t:\sigma(t)<t\}=r-\operatorname{asc}(\sigma)$, since $\sigma$ has no fixed points. Substituting $u=\sigma^{-1}(t)$ in the second gives
\begin{equation}
  \#\{t:\sigma^{-1}(t)>t\}
  =\#\{u:\sigma(u)<u\}
  =r-\operatorname{asc}(\sigma).
\end{equation}
Therefore $\sum_{t=1}^r e_t=2\bigl(r-\operatorname{asc}(\sigma)\bigr)$. Iterating \cref{eq:app-partial-contraction-recurrence} proves
\begin{equation}
  \norm{\mathcal O_\sigma}_\infty
  \le\prod_{t=1}^r\norm{L_t}_\infty
  \le 2^r n^{\frac12\sum_t (e_t-1)}
  =2^r n^{s-\operatorname{asc}(\sigma)}. \qedhere
\end{equation}
\end{proof}

\subsection{Completion of the proof}

\begin{proof}[Proof of \Cref{lem:noncrossing-pair}]

By \cref{eq:app-a-limit-error}, $c_\sigma=O_r(n^{-k(\sigma)})$.
Using \cref{lem:closed-diagram-bound} and $s-k(\sigma)=b(\sigma)$, we have
\begin{equation}\label{eq:app-main-term-bound}
  \norm{c_\sigma\mathcal O_\sigma}_\infty
  =O_r\bigl(n^{-k(\sigma)}\bigr)\cdot 2^r n^{s-\operatorname{asc}(\sigma)}
  =O_r\bigl(n^{\,b(\sigma)-\operatorname{asc}(\sigma)}\bigr)
  =O_r(1).
\end{equation}
Every cycle of a fixed-point-free permutation $\sigma\in \mathfrak S_r$ contains at least one ascent $t<\sigma(t)$ (for instance, its smallest element), and every orbit of $\sigma$ and $\pi$ contains at least one cycle of $\sigma$. Hence
\begin{equation}
  \operatorname{asc}(\sigma)\ge\#\sigma\ge b(\sigma).
\end{equation}
Therefore, in fact, $\norm{c_\sigma\mathcal O_\sigma}_\infty=O_r(n^{-1})$ unless $\operatorname{asc}(\sigma)=\#\sigma=b(\sigma)$.
This holds if and only if:
\begin{enumerate}
\item $\#\sigma=b(\sigma)$, meaning each cycle of $\sigma$ is a union of pairs of $\pi$;
\item each cycle of $\sigma$ has a single ascent, meaning its elements appear in strictly decreasing cyclic order: $(t_1\,t_2\,\cdots\,t_\ell)$ with $t_1>t_2>\cdots>t_\ell$.
\end{enumerate}
Hence $\sigma$ is uniquely determined by $\pi$ together with a set partition of the pairs of $\pi$.
Indeed, by (1), the cycles of $\sigma$ amount to a grouping of the $s$ pairs of $\pi$ into blocks. By (2), the order in each cycle is forced.

We claim that $a_\sigma^\infty=\mathbf 1_{\{\sigma=\pi\}}$ if $\sigma\in\mathfrak S_r$ satisfies (1) and (2).
First we prove
\begin{equation}\label{eq:app-d-pi-sigma}
  d(\pi,\sigma)=s-b(\sigma)=k(\sigma).
\end{equation}
Recall that $\rho\le\gamma$ in $\mathfrak S_r$ if 
Recall Biane's theorem \cite[Theorem~1]{Bia97}, $\rho\le\gamma$ (defined as $\abs{\rho}+\abs{\rho^{-1}\gamma}=\abs{\gamma}$, namely, a geodesic between $\gamma$ and $e$ goes through $\rho$) if and only if on each cycle of $\gamma$, the cycles of $\rho$ form a noncrossing partition and inherit the cyclic order of $\gamma$. 
In particular, here we have $\pi\le\sigma$. 
Since $\abs{\pi}=r-s$ and $\abs{\sigma}=r-b(\sigma)$, this gives $d(\pi,\sigma)=\abs{\sigma}-\abs{\pi}=s-b(\sigma)=k(\sigma)$, proving \cref{eq:app-d-pi-sigma}.

Next we identify the permutations $\tau$ contributing to $a_\sigma^\infty$ in \cref{eq:app-a-limit-error}: for every $\tau\in\mathfrak S_r$,
\begin{equation}\label{eq:app-geodesic-interval}
  d(\tau,\sigma)=k(\sigma)-k(\tau)
  \iff
  \tau\in[\pi,\sigma]
  :=\bigl\{\tau\in \mathfrak S_r:\ d(\pi,\tau)+d(\tau,\sigma)=d(\pi,\sigma)\bigr\}.
\end{equation}
If $\tau\in[\pi,\sigma]$, then $d(\pi,\tau)+d(\tau,\sigma)=k(\sigma)$ by \cref{eq:app-d-pi-sigma}; since $d(\pi,\tau)\ge k(\tau)$ and $d(\tau,\sigma)\ge k(\sigma)-k(\tau)$ by \cref{eq:app-k-lipschitz}, both inequalities must be equalities. Conversely, suppose $d(\tau,\sigma)=k(\sigma)-k(\tau)=b(\tau)-b(\sigma)$. Along a shortest transposition path of length $m=d(\tau,\sigma)$ from $\sigma$ to $\tau$, each step changes $b(\cdot)$ by at most $1$, and can increase $b(\cdot)$ only when it splits a cycle. To increase $b(\cdot)$ by $m$ in $m$ steps, every step must split a cycle, so $\#\tau=\#\sigma+m$. Together with $\#\sigma=b(\sigma)$ and $b(\tau)=b(\sigma)+m$, this gives $\#\tau=b(\tau)$, so $\tau$ satisfies (1). Moreover, $\#\tau=\#\sigma+d(\tau,\sigma)$ means $\tau\le\sigma$; hence each cycle of $\tau$ inherits the decreasing cyclic order from $\sigma$, so $\tau$ also satisfies (2). Applying \cref{eq:app-d-pi-sigma} to $\tau$ gives $d(\pi,\tau)=k(\tau)$ and thus $\tau\in[\pi,\sigma]$.

Finally, we evaluate $a_\sigma^\infty$. Substituting $\rho=\pi^{-1}\tau$ and $\gamma=\pi^{-1}\sigma$ identifies $[\pi,\sigma]$ with the poset $[e,\gamma]=\{\rho\in\mathfrak S_r:\rho\le\gamma\}$. Moreover, $(\rho,\gamma)\mapsto\operatorname{Mob}(\rho^{-1}\gamma)$ is the M\"obius function of the poset $(\mathfrak S_r,\le)$ \cite{CS06,NS06}, and therefore satisfies $\sum_{\rho\in[e,\gamma]}\operatorname{Mob}(\rho^{-1}\gamma)=\delta_{e,\gamma}$ by definition. 
Together with \cref{eq:app-a-limit-error,eq:app-geodesic-interval}, this gives
\begin{equation}
  a_\sigma^\infty
  =\sum_{\tau\in[\pi,\sigma]} \operatorname{Mob}(\tau^{-1}\sigma)
  =\sum_{\rho\in[e,\gamma]} \operatorname{Mob}(\rho^{-1}\gamma)
  =\delta_{e,\gamma}
  =\mathbf 1_{\{\sigma=\pi\}}.
\end{equation}

Thus $\norm{c_\sigma\mathcal O_\sigma}_\infty=O_r(n^{-1})$ for all $\sigma\ne\pi$, while $c_\pi=1+O_r(n^{-1})$ (since $k(\pi)=0$) and $\norm{\mathcal O_\pi}_\infty\le 2^r$ by \cref{lem:closed-diagram-bound} (in fact $\norm{\mathcal O_\pi}_\infty\le 1$, since every step in its proof has $e_t=1$ when $\sigma=\pi$). Substituting this into \cref{eq:app-word-expansion} gives
\begin{equation}
  \norm{\mathbb E\bigl[\mathcal W_\pi(T_1,\ldots,T_r)\bigr]-\mathcal O_\pi}_\infty
  =O_r(n^{-1}).
\end{equation}
Finally, by \cref{eq:app-Y-properties} and the definition of $\mathcal E_\pi$, we have $\mathcal O_\pi=\mathcal E_\pi(T_1,\ldots,T_r)$, proving \cref{eq:noncrossing-lemma-bound}.
\end{proof}

\bibliographystyle{alpha}
\bibliography{data-hiding}

\end{document}